\pdfoutput=1
\documentclass[a4paper, 11pt,dvipsnames]{article} 

\usepackage[utf8]{inputenc}
\usepackage{amsthm}
\usepackage{comment}
\usepackage{amssymb}

\usepackage{complexity}
\usepackage{amsthm}
\usepackage{comment}
\usepackage{amssymb}

\usepackage[fleqn]{mathtools}
\usepackage{thmtools}
\usepackage{thm-restate}
\usepackage{xcolor}
\usepackage{xspace}
\usepackage{hyperref}
\usepackage{enumerate}
\hypersetup{
    colorlinks=true,     
    citecolor=blue,      linkcolor=blue,      }

\usepackage{subcaption}
\usepackage{microtype}
\usepackage{mathpazo}
\usepackage[left=1in,right=1in,bottom=1in,top=1in]{geometry}
\hypersetup{colorlinks,linkcolor=[rgb]{0.3,0.3,0.6},   citecolor=[rgb]{0.2, 0.6, 0.2},urlcolor=[rgb]{0.6, 0.2, 0.2}}\date{}

\usepackage{graphicx}
\usepackage{tikz}
\usepackage[lined,boxed]{algorithm2e}
\usepackage{booktabs}
\newcommand{\splitatcommas}[1]{\begingroup
  \begingroup\lccode`~=`, \lowercase{\endgroup
    \edef~{\mathchar\the\mathcode`, \penalty0 \noexpand\hspace{0pt plus 1em}}}\mathcode`,="8000 #1\endgroup
}

\SetKwInOut{Input}{Input}
\SetKwInOut{Output}{Output}
\newtheorem{problem}{Problem}
\newtheorem{theorem}{Theorem}

\newtheorem{proposition}{Proposition}  
\newtheorem{definition}{Definition}
\newtheorem{lemma}[theorem]{Lemma}
\newtheorem{corollary}{Corollary}
\global\long\def\lst#1#2{#1_{1},#1_{2},\dots,#1_{#2}}

\global\long\def\lstp#1#2{#1_{1}+#1_{2}+\dots+#1_{#2}}
\global\long\def\lstl#1#2{#1_{1}<#1_{2}<\dots<#1_{#2}}

\global\long\def\lsst#1#2{\splitatcommas{\sqrt{#1_{1}},\sqrt{#1_{2}},\dots,\sqrt{#1_{#2}}}}

\makeatletter
\DeclareRobustCommand\onedot{\futurelet\@let@token\@onedot}
\def\@onedot{\ifx\@let@token.\else.\null\fi\xspace}
\def\ie{\emph{i.e}\onedot} 
\makeatother

\DeclareMathOperator{\spn}{span}
\newcommand{\iu}{{i\mkern1mu}}

\newcommand\abs[1]{\left|#1\right|}
\newcommand\parc[1]{\left(#1\right)}
\newcommand{\OF}{\mathcal{O}}
\newcommand{\ord}{\operatorname{ord}}
\newcommand{\posslp}{\operatorname{PosSLP}}

\newcommand{\SSR}{\operatorname{SSR}}
\newcommand{\SSReq}{\operatorname{SSR}_{\operatorname{eq}}}
\newcommand{\SSRslp}{\operatorname{SSR}_{\operatorname{SLP}}}
\newcommand{\equslp}{\operatorname{EquSLP}}
\newcommand{\pit}{\operatorname{PIT}}
\newcommand{\eqdef}{\vcentcolon =}

\newcommand{\BK}{\mathbb{K}}
\newcommand{\Z}{\mathbb{Z}}

\usepackage[nameinlink]{cleveref}
\crefname{definition}{Definition}{Definitions}
\crefname{theorem}{Theorem}{Theorems}
\crefname{lemma}{Lemma}{Lemmas}
\crefformat{enumi}{Part (#1)}
\crefname{problem}{Problem}{Problems}
\crefname{equation}{Equation}{Equations}
\crefname{proposition}{Proposition}{Propositions}
\crefname{conjecture}{Conjecture}{Conjectures}
\crefname{corollary}{Corollary}{Corollaries}
\crefname{appendix}{Appendix}{Appendices}
\crefname{algorithm}{Algorithm}{Algorithms}
\crefname{section}{Section}{Sections}
\crefname{figure}{Figure}{Figures}
 \providecommand{\keywords}[1]
{\noindent \textbf{Keywords } #1
}
\usepackage{amsfonts}
\usepackage{amsthm} 

\author{Louis Gaillard \thanks{{\'E}cole Normale Sup{\'e}rieure de Lyon. Email: \texttt{louis.gaillard@ens-lyon.fr}} \and Gorav Jindal \thanks{Max Planck Institute for Software Systems. Email: \texttt{gjindal@mpi-sws.org}}}
\begin{document}

  \title{On the Order of Power Series and the Sum of Square Roots Problem}

\maketitle
\begin{abstract}

  This  paper focuses on the study of the order of power series that are linear combinations of a given finite set of power series. The order of a formal power series, known as $\ord(f)$, is defined as the minimum exponent of $x$ that has a non-zero coefficient in $f(x)$. Our first result is that the order of the Wronskian of these power series is equivalent up to a polynomial factor, to the maximum order which occurs in the linear combination of these power series. This implies that the Wronskian approach used in (Kayal and Saha, TOCT'2012) to upper bound the order of sum of square roots is optimal up to a polynomial blowup. We also demonstrate similar upper bounds, similar to those of (Kayal and Saha, TOCT'2012), for the order of power series in a variety of other scenarios. We also solve a special case of the inequality testing problem outlined in (Etessami et al., TOCT'2014).

  In the second part of the paper, we study the equality variant of the sum of square roots problem, which is decidable in polynomial time due to (Bl\"omer, FOCS'1991). We investigate a natural generalization of this problem when the input integers are given as straight line programs. Under the assumption of the Generalized Riemann Hypothesis (GRH), we show that this problem can be reduced to the so-called ``one dimensional'' variant. We identify the key mathematical challenges for solving this ``one dimensional'' variant.
\end{abstract}
\keywords{Formal power series, Sum of square roots, Wronskian, Differential equations, Straight line Programs.}

\noindent \textbf{2012 ACM Subject Classification }Theory of computation $\to$ Algebraic complexity theory, Mathematics of computing $\to$ Probabilistic algorithms.
 \maketitle

\section{Introduction}

For numerous decision problems that require determining the sign of expressions with real numbers, their complexity class (e.g., if they belong to $\P$ or not) is often unknown. A notable instance is the Sum of Square Roots problem, which can be described as:	

\begin{problem} [Sum of Square Roots ($\SSR$)] Given a list $(\lst a n)$ of positive integers  and $(\lst \delta n) \in \{-1,1\}^n$, decide whether
$\sum_{i=1}^n \delta_i \sqrt{a_i} > 0$.  \label{pb:SSR} \end{problem}

The complexity of this problem has been extensively investigated and remains an open question according to Garey, Graham, and Johnson~\cite{garey1976some}. Additionally, it has been hypothesized that it lies in $\P$, as proposed by Malajovich in 2001~\cite{Malajovich2001AnEV}. The $\SSR$ problem shares deep connections with classical geometric problems such as the Euclidean Traveling Salesman Problem (ETSP),  which is not known to be in $\NP$. It is readily seen to be in $\NP$ relative to an oracle of $\SSR$. An important related problem where the task is to determine the sign of an integer (encoded by a straight line program), is the so called $\posslp$ problem. A sequence $P$ of integers $(b_0, \lst b \ell)$ is said
to be a straight line program (SLP) if $b_{0}=1$ and for all $1\leq i\leq\ell$, $b_{i}=b_{j}\circ_i b_{k}$, where $\circ_i\in\{+,-,*\}$ and $j,k<i$. We say that this SLP $P$ computes the integer  $b_{\ell}$. We say that $\ell$ is the size (or length) of this SLP $P$.
\begin{problem} [$\posslp$] Given an SLP $P$, decide if the integer $n_P$ computed by $P$ is positive. \label{pb:posslp}
\end{problem}
The approximation of $\sqrt{a_i}$ in \cref{pb:SSR} to an appropriate precision leads to the reduction of $\SSR$ to $\posslp$, as demonstrated in \cite{comp_num_analysis}. $\posslp$ was introduced in this work to bridge the gap between classical models of computation and computation over the reals in the Blum-Shub-Smale (BSS) model \cite{SmaleRealCompu1997}, which is a widely studied model for the study of computational complexity in numerical analysis.

It was shown in \cite{comp_num_analysis} that $\posslp$ captures the ``Boolean part'' of languages decidable in the BSS model in polynomial time, with polynomial time Turing reductions. Additionally, $\posslp$ was proven to lie within the counting hierarchy $\CH$, implying that $\SSR$ also falls within $\CH$. To date, this represents the best known upper bound for the complexity of $\SSR$.

Our contributions can be found in~\cref{sec:our_results} after having introduced some useful previous work on ~\cref{pb:SSR} in~\cref{sec:relatedwork}.
\subsection{Related work}\label{sec:relatedwork}

In different scenarios, there is a need to determine if the sum of square roots is equal to zero or not, rather than determining the sign of the expression. As a result, we encounter an intriguing problem known as $\SSReq$.

\begin{problem}[$\SSReq$] Given a list $(\lst a n)$ of integers and $(\lst \delta n) \in \{-1,1\}^n$, decide whether $\sum_{i=1}^n \delta_i \sqrt{a_i} = 0$.
\label{pb:SSReq} 
\end{problem}
Bl{\"o}mer~\cite{blomer1993computing} gave a polynomial time algorithm to compute sums of radicals and also proved that $\SSReq$ is in $\P$. Thus, testing if a signed sum
of square roots is zero seems to be easier than deciding its sign.  A natural idea for an algorithm for the $\SSR$ problem would be to approximate the sum with
a floating point number thanks to classical numerical algorithms. However, one would need a result on the required number of bits of precision of the
approximation to ensure that the two signs coincide. It is known (\cite{brent1976fast}) that for an integer of bit size at most $B$, its square root can be
approximated up to $m$ bits of precision in time $\mathrm{poly}(B,m)$. And this implies that a solution of the following number theoretic problem would lead to a
polynomial time algorithm for the $\SSR$ problem.
\begin{problem} [Lower-bound on a nonzero sum of square roots] Given a sum $S = \sum_{i=1}^n \delta_i \sqrt{a_i}$, with $\delta_i \in \{-1,1\}$ and $1 \le a_i <
  2^B$, can we find a polynomial $q(n,B)$ such that \begin{align*} S \neq 0 \implies \abs{S} \ge \frac{1}{2^{q(n,B)}}?  \end{align*}
\label{pb:SSR_lower_bd} 
\end{problem}
In contrast to \cref{pb:SSR_lower_bd}, one can also try to find $a_i$'s and corresponding $\delta_i$'s such that the absolute value of  $\sum_{i=1}^n \delta_i \sqrt{a_i}$  is \emph{small}. In this direction, it was shown in \cite{QIAN2006194} that $\abs{\sum_{i=0}^m\binom{m}{i} (-1)^i\sqrt{n+i}}=O(n^{-m+\frac{1}{2}})$. Kayal and Saha~\cite{kayal2012sum} chose to approach \cref{pb:SSR_lower_bd} by formulating a related question over polynomials. This approach proved to be simpler. They focused on non-zero sums of square roots of polynomials, which they viewed as power series, and demonstrated that the valuation (or order) of such a series cannot be \emph{too high}.

\begin{theorem}[Bounding the order of sum of square roots of polynomials~\cite{kayal2012sum}] \label{thm:order_SSR}
For $1 \le i \le n$, let $c_i \in \mathbb{C}$ and  $f_i, g_i \in \mathbb{C}[x]$ of degree at most $d$ with $f_i(0) \neq  0$. We denote and fix
$\sqrt{f_i(x)} \in \mathbb{C}[[x]]$ one of the two square roots of $f_i(x)$. If the sum $S(x) = \sum_{i=1}^{n} c_i g_i(x) \sqrt{f_i(x)} $ is non-zero, then $\ord(S)\leq dn^2+n$.
\end{theorem}
The main technical argument of this proof is the study of the Wronskian determinant of the family $(g_i\sqrt{f_i})_{1 \le i \le n}$, because by Cramer's rule,
one can easily bound the order of $S$ with respect to the order of this Wronskian. 

Next, they apply this result to a set of integers by representing them as polynomials and they confirm that the solution to ~\cref{pb:SSR_lower_bd} is affirmative for a significant subclass of integers known as \emph{polynomial integers}.

\begin{theorem}[$\SSR$ for \emph{polynomial integers}~\cite{kayal2012sum}] Suppose \\ $S = \sum_{i=1}^{n} \delta_i \sqrt{a_i}$ is non-zero with $\delta_i \in \{-1,1\}$, such that
  every positive integer $a_i$ is of the form $a_i = X^{d_i} + b_{1,i}X^{d_i -1} + \dots + b_{d_i,i}$ with $d_i >0$, $X$ a positive integer and $d_i, b_{j,i}$ are
  integers. Let $B = \max \left( \left\{\left| b_{j,i} \right| \right\}_{j,i}, 1 \right)$ and $d = \max_i d_i$. 
  There exist two fixed integer polynomials $p(n,d), q(n,d)$ in $n$ and $d$ such that if $X \ge (B+1)^{p(n,d)}$, then $S$ is lower bounded as $\abs{S} \ge \frac{1}{X^{q(n,d)}}$. 
\label{thm:SSR_PolynomialInteger} 
\end{theorem}

Based on these observations, it raises the question of whether the concept of bounding the order of sums of power series can be generalized to other families beyond square roots. This could potentially yield significant insights into the complexity of determining the positivity of expressions involving irrational real numbers.

  \subsection{Our results}
  \label{sec:our_results}
  From now on, $\BK$ denotes a field of characteristic 0. All the logarithms in this paper are natural logarithms with base $e$, unless otherwise stated. We now define some measures to formally state our results.
    \begin{definition} Suppose $\mathcal{F} \subseteq \BK[[x]]$ is a finite dimensional linear subspace of $\BK[[x]]$, we define \begin{align*}
    \OF(\mathcal{F})\eqdef \sup \{ \ord(f) \left| f \in \mathcal{F} \setminus \{0\}  \right \} \in \mathbb{N} \cup \infty.  \end{align*} And for $\mathbf{f} =
    (\lst f n)$ a family of linearly independent power series, we  define \begin{align*} \OF(\lst f n) = \OF(\mathbf{f}) \eqdef \OF( \spn (\lst f n)).
    \end{align*} \end{definition} A set $\{\lst f n\}$  of $n$ functions over $\BK$ is said to be linearly dependent if there exist scalars $\lst c n \in \BK$ (not all zero)
    such that $\sum_{i=1}^{n}c_if_i$ is zero. To define the Wronskian of $\{\lst f n\}$, we assume that each $f_i$ is $n-1$ times differentiable.  The Wronskian of
    $\{\lst f n\}$, denoted $W(\lst f n)$ is defined as the determinant of the following matrix: \begin{align*} &W(\mathbf{f})= W(\lst f n) \eqdef \det\begin{pmatrix} f_1 & \dots &
      f_n \\ f_1^{(1)} & \dots & f_n^{(1)} \\ \vdots & \vdots & \vdots \\ f_1^{(n-1)}  & \dots & f_n^{(n-1)} \end{pmatrix}  \end{align*}
      \begin{proposition} Let $\mathcal{F} \subseteq \BK[[x]]$ be an $n$-dimensional linear subspace. For any basis $(\lst f n)$  of $\mathcal{F}$,  $\ord (W(\lst f n))$ does not depend on the choice of the basis $(\lst f n)$. We denote this quantity by
    $W_{\ord}(\mathcal{F})$.  \label{prop:ord_Wr_subspace} \end{proposition}
\begin{proof} Let $(\lst g n)$ be another basis of $\mathcal{F}$. There exists an invertible $n \times n$ matrix $A$ with entries in $\BK$ such that
        \begin{align*} \begin{bmatrix} g_1 & \dots & g_n \end{bmatrix} = \begin{bmatrix} f_1 & \dots & f_n \end{bmatrix} \cdot A.  \end{align*} By linearity of
        the differentiation, we also have for all $0 \le j < n$, \begin{align*} \begin{bmatrix} g_1^{(j)} & \dots & g_n^{(j)} \end{bmatrix} = \begin{bmatrix}
      f_1^{(j)} & \dots & f_n^{(j)} \end{bmatrix} \cdot A.  \end{align*} Thus we have $W(\lst g n) = W(\lst f n) \cdot \det(A)$. By using the fact that $\det
    (A) \in \BK^*$, the result follows.  \end{proof}
The following theorem  bounds $\OF(\mathcal{F})$ in terms of $W_{\ord}(\mathcal{F})$. This theorem is actually a result of Voorhoeve and Van der Poorten~\cite{voorhoeve1975wronskian}. The same idea is used by Kayal and Saha~\cite{kayal2012sum} to establish the bound in~\cref{thm:order_SSR}.
\begin{theorem}[Theorem~1  in~\cite{voorhoeve1975wronskian}]\label{thm:upbd_wronS} 
Let $\mathcal{F}$ be an $n$-dimensional linear subspace of $\BK[[x]]$. Then, $\OF(\mathcal{F})   \le W_{\ord}(\mathcal{F}) + n-1$.
\end{theorem}
  In~\cref{sec:wronskian}, we show that \cref{thm:upbd_wronS}  is almost tight, in the sense that the order of the Wronskian of a family $(\lst h n)$ of power series is equivalent, up to a polynomial factor, to the maximum order of a non-zero linear combination of the $h_i$'s. This result is formalized in~\cref{thm:lowbd_wronS} below.

  \begin{restatable}{theorem}{OptimalWronskian}\label{thm:lowbd_wronS}
  Let $\mathcal{F}$ be an $n$-dimensional linear subspace of $\BK[[x]]$. Then, $W_{\ord}(\mathcal{F}) \le n \cdot \OF(\mathcal{F} ) - \binom{n}{2}$. 
\end{restatable}

For a full proof of  \cref{thm:lowbd_wronS}, refer to \cref{sec:wronskian}. We then demonstrate that the approach of Kayal and Saha in \cref{thm:order_SSR} can be extended to sums of solutions of linear differential equations of order 1 with polynomial coefficients, resulting in a comparable bound on the sum's order. This result is formally stated in \cref{thm:diffeq1} below.
  \begin{restatable}{theorem}{diffeqorderboundthm} 
  Let $S(x) = \sum_{i=1}^n c_i g_i(x) y_i(x)$, with $y_i' - \frac{p_i}{q_i}y_i = 0 $, where $c_i \in \BK$ and $g_i, p_i, q_i \in \BK[x]$ of degree at most $d$. We assume that $q_i(0) \neq 0$ and that each $y_i \in \BK[[x]]$. If $S \neq 0$, then
  \begin{align*}
    \ord(S) \le \sum_{i=1}^n \ord(y_i) + n^2 d + n-1.  
  \end{align*} 
  \label{thm:diffeq1} 
\end{restatable}
  \paragraph{Proof idea for \cref{thm:diffeq1}} 
  The bound for $\ord(S)$ follows from~\cref{thm:upbd_wronS}, once we have bounded the order of the Wronskian of the family $(g_iy_i)_{1 \le i \le n}$ by $\sum_i \ord (y_i) + n^2d$. To do this, we study the entries of the Wronskian matrix and use the differential equations to replace derivatives of the $y_i$'s. We show that every entry $(i,j)$ of the Wronskian matrix can be written $y_i \frac{m_{i,j}}{q_i^{n-1}}$, where $q_i$ is the denominator in the differential equation satisfied by $y_i$ and $m_{i,j}$ is a polynomial of degree $\le nd$. 
  The result directly follows by bounding the order of the polynomials $m_{i,j}$ by their degree.
For a full proof of  \cref{thm:diffeq1}, see \cref{sec:diffeqs}.

In \cref{sec:sum_sqroot_slp}, we investigate the following \cref{pb:SSRSLP} ($\SSRslp$) that can be regarded as a generalization of~\cref{pb:SSReq} ($\SSReq$) in the context of SLPs.
\begin{problem} [$\SSRslp$] 
\label{pb:SSRSLP} 
  Given as input $n$ straight line programs $(\lst P n)$ of size $\le s$ and $(\lst \delta n) \in \{-1,1\}^n$, such that $P_i$ computes the positive integer $a_i$, decide whether $\sum_{i=1}^{n} \delta_i \sqrt{a_i} = 0$.  
\end{problem}

Here we show that assuming GRH, under randomized polynomial time Turing reductions, $\SSRslp$ can be reduced to the following ``one-dimensional'' variant of $\SSRslp$, 
\begin{problem} [One dimensional $\SSRslp$] \label{pb:1dimSSR} Given $n$ straight line programs computing $n$ positive integers $(\lst a n)$ and $(\lst \delta n) \in \{-1,1\}^n$, with the promise that $\dim (\spn_{\mathbb{Q}}(\lsst a n)) = 1$.  Decide if $\sum_{i=1}^{n} \delta_i \sqrt{a_i} = 0$.
\end{problem}

\paragraph{Proof idea for reducing $\SSRslp$ to one dimensional $\SSRslp$}
Given an instance of \cref{pb:SSRSLP}, we separate the inputs into several one-dimensional subgroups. This is possible because we demonstrate an efficient randomized algorithm to test the linear dependency of square roots of integers given by SLPs, using the ideas of Kneser \cite{Kneser1975} and recounted out in Blömer's work in~\cite{blomer1993computing}. According to Kneser's result \cite{Kneser1975}, a set of square roots are linearly dependent over $\mathbb{Q}$ if and only if there exists a pair of linearly dependent square roots within the set. Once the instance of \cref{pb:SSRSLP} is separated into one-dimensional subgroups, it remains only to determine if each subsum is zero, which can be done using an oracle for \cref{pb:1dimSSR}. For a full proof, see~\cref{sec:sum_sqroot_slp}.

In~\cref{sec:linearcomblogs}, we show a similar upper bound to that of \cref{thm:order_SSR}, on the order of sums of logarithms of real polynomials. 
\begin{restatable}{proposition}{sumoflogsorderboundprop}
   Let $S(x) = \sum_{i=1}^n c_i  \log(f_i(x)) \neq 0$, where $c_i \in \mathbb{R}$, $f_i \in \mathbb{R}[x]$  of degree at most $d$ and $f_i(0) > 0$. Then $ \ord(S) \le nd$.  \label{prop:log} 
\end{restatable}

We also show an analogous result to~\cref{thm:SSR_PolynomialInteger} but for the problem of the positivity testing of linear forms of
  logarithms of integers whose complexity is connected to deep conjectures in number-theory~\cite{etessami2014note}, such as a refinement of the
  $abc$ conjecture formulated by Baker~\cite{baker1998logarithmic}. This result essentially follows from \cref{thm:sumlog} below, which is our analogue of~\cref{thm:SSR_PolynomialInteger}. For a full proof of \cref{thm:sumlog}, see~\cref{sec:linearcomblogs}.

  \begin{restatable}[Sum of logarithms of \emph{polynomial integers}]{theorem}{sumoflogsofpolyintsthm}  \label{thm:sumlog} 
  Suppose $E = \sum_{i=1}^n c_i \log{a_i}$ is non-zero, where $c_i \in \mathbb{Z}$, and every $a_i$ is a
  positive integer of the form $a_i = X^{d_i} + b_{1,i}X^{d_i - 1 } + \dots + b_{d_i,i}$, where $d_i > 0$, $X$ is a positive integer and $d_i,b_{j,i}$ are
  integers. Let $B = \max \left( \left\{\abs{b_{j,i}} \right\}_{j,i}, 1 \right), d = \max_i d_i + 1$, and $C = \max_i \left| c_i \right|$. 
    There exist two fixed integer polynomial $p_1(n,d), p_2(n,d)$ in $n$ and $d$ such that if $X > \max \left( C^2, (B+1)^{p_1(n,d)} \right)$, then $E$ is lower bounded as $\abs{E} \ge \frac{1}{X^{p_2(n,d)}}$.
\end{restatable}

  \section{Order of the Wronskian determinant} \label{sec:wronskian}

  In this section we show that for a family of power series, the maximal order that can occur in a non-zero linear combination is the order of the Wronskian determinant of the family up to a polynomial factor.	Our proof of~\cref{thm:lowbd_wronS} is inspired from the ideas developed in~\cite{bostan2010wronskians}. For $d, k \in \mathbb{N}$, we denote \begin{align*} (d)_{k} \eqdef d(d-1)\dots (d-k+1), \end{align*} with the convention $(d)_0 = 1$.
\begin{definition} [Vandermonde determinant] Let $(\lst d n) \in \BK^{n}$, we define the corresponding Vandermonde determinant as follows: \begin{align*} V(\lst
  d n) \eqdef \det 
\left((d_j^{i-1})_{1 \le i,j \le n }\right) = \prod_{1 \le i < j \le n} (d_j - d_i).  \end{align*} \end{definition} 
\begin{lemma} [Wronskian of monomials, Lemma~1 in~\cite{bostan2010wronskians}] The Wronskian of the monomials $a_1x^{d_1}, \dots , a_n x^{d_n}$ is
  \begin{align*} W(a_1x^{d_1}, \dots , a_n x^{d_n}) = V(\lst d n) \left(\prod_{i=1}^n a_i \right) x^{d_1 + \dots + d_n - \binom{n}{2}}.  \end{align*}
\label{lemma:wronmonom} \end{lemma}

\begin{lemma} [Lemma~2 in~\cite{bostan2010wronskians}] Let $\lst f n$ be a family of   $\BK[[x]]$ which are linearly independent over $\BK$. There
  exists an invertible $n \times n$ matrix $A$ with entries in $\BK$ such that the power series $\lst g n$ defined by \begin{align} \begin{bmatrix} g_1 & \dots
  & g_n \end{bmatrix} = \begin{bmatrix} f_1 & \dots & f_n \end{bmatrix} \cdot A \label{eq:distinctorder} \end{align} are all non-zero and have mutually distinct
  orders.  \label{lemma:distinctorder} \end{lemma}

  \begin{lemma} [Wronskian of distinct order power series] If the non-zero power series $\lst g n \in \BK[[x]]$ have mutually distinct orders $\lst d n$, then
    their Wronskian $W(\lst g n)$ is non-zero and satisfies: \begin{align*} \ord (W(\lst g n)) = \sum_{i=1}^n d_i - \binom{n}{2}.  \end{align*}
  \label{lemma:ord_wr_dist_order} \end{lemma}
\begin{proof} If the $g_j$'s are monomials, \ie $g_j = a_j x^{d_j} $, the result is a direct consequence of \cref{lemma:wronmonom}, and in this case, the $(i,j)$ entry of the Wronskian matrix is $w_{i,j} = a_j (d_j)_{i-1}x^{d_j-i+1} $.  In the general case, let $g_j = a_j x^{d_j} + u_j$
    with $u_j \in \BK[[x]]$ of order $> d_j$. Then the $(i,j)$ entry of the Wronskian matrix now becomes $w_{i,j} \times (1 + x r_{i,j})$ for some $r_{i,j} \in \BK[[x]]$.

    So we have $W(\lst g n)= a_1 \dots a_n x^{d_1 + \dots + d_n - \binom{n}{2}}\det(D) $ where $D$ is the $n \times n$ matrix $D= ((d_j)_{i-1}(1 + x r_{i,j}))_{1 \le i,j \le n}$.
Now we evaluate $D$ at $x=0$, and we obtain $D(0) = ((d_j)_{i-1})_{1 \le i,j \le n}$. With elementary row operations (which preserve the determinant), as in the proof of~\cref{lemma:wronmonom} in~\cite{bostan2010wronskians}, we can transform $D(0)$ into the Vandermonde matrix associated to $d_1, \dots, d_n$. So, $\det(D(0)) = V(\lst d n) \neq 0$, because the $d_i$'s are distinct. Thus $\det(D)$ is non zero modulo $x$, so $\det D$ has order zero and the result follows. 
   \end{proof}   
We now formulate a tight variant (\cref{thm:lowbdtight_wronS}) of \cref{thm:lowbd_wronS}, which immediately implies \cref{thm:lowbd_wronS}. Suppose $ \mathcal{F} \subseteq \BK[[x]]$ is a finite dimensional subspace of $\BK[[x]]$. ~\cref{lemma:distinctorder} shows that if $\dim(\mathcal{F})=n$ then there exist $n$ power series $\lst g n$ in $\mathcal{F}$ which have distinct orders $\lst d n$. Moreover $\lst g n$ also form a basis of $ \mathcal{F}$. We now claim that these are the only possible orders of any power series in $\mathcal{F}$. Assume $\lstl d n$. If $f=\sum_{i=1}^{n}\lambda_i g_i$ (with $\lambda_i \in \BK$) is a non-zero power series in $\mathcal{F}$, then it is clear that $\ord(f)=d_j$ where $j$ is the minimum index such that $\lambda_j \neq 0$. With this claim, we formulate the following tight variant of \cref{thm:lowbd_wronS}.

  \begin{theorem}\label{thm:lowbdtight_wronS}
    Let $\mathcal{F}$ be an $n$-dimensional linear subspace of $\BK[[x]]$. Suppose  $\lst d n$ are the distinct orders of power series which occur in  $\mathcal{F}$. Then, $W_{\ord}(\mathcal{F}) = \sum_{i=1}^{n}d_i - \binom{n}{2}$. 
\end{theorem}
    \begin{proof} 
By \cref{lemma:distinctorder}, there exists a basis $(\lst g n)$ of $\mathcal{F}$ with distinct order $\lst d n$. By using \cref{lemma:ord_wr_dist_order}, $ 	W_{\ord}(\mathcal{F}) = \sum_{i=1}^{n}d_i  - \binom{n}{2}$.\end{proof} 

	\OptimalWronskian*
  \begin{proof} The claim immediately follows from~\cref{thm:lowbdtight_wronS}. \end{proof}
This completes the proof of \cref{thm:lowbd_wronS}. By combining \cref{thm:lowbd_wronS} and \cref{thm:upbd_wronS}, we conclude that the order of the Wronskian determines the maximum order of linear combinations of power series, up to a polynomial factor in $n$. 

Now we show an example where the bound claimed in \cref{thm:upbd_wronS} is  tight. Suppose $\mathcal{G}$ is the linear subspace of  $\BK[[x]]$ generated by the monomials $1,x,x^2,\dots,x^{n-1}$. It is easy to see that Wronskian of $1,x,x^2,\dots,x^{n-1}$ is $\prod_{i=1}^{n-1}i!$. Hence  $W_{\ord}(\mathcal{G})=0$. We also have that $ \OF(\mathcal{G})=n-1$. \section{Sums of solutions of linear differential equations}\label{sec:diffeqs}	

As a generalization of \cref{thm:order_SSR} and as an application of~\cref{thm:upbd_wronS}, we now prove a polynomial bound on the order of a sum of power series
that are solutions of linear differential equations of order 1 with polynomial coefficients.

	\diffeqorderboundthm*
  \begin{proof}
Let $h_i = g_iy_i$. Without loss of generality, we can assume that $\mathbf{h}=(\lst h n)$ are linearly independent. If they are not, we can
    rewrite $S$ as a linear combination of a subfamily of the $h_i$'s that are linearly independent. 

    We shall bound the order of the Wronskian $W(\lst h n)$ and apply \cref{thm:upbd_wronS}.  We have: \begin{align*} y_i' &= \frac{p_i}{q_i}y_i, \\ y_i'' &=
      \frac{p_i'q_i - p_iq_i'}{q_i^2}y_i + \frac{p_i}{q_i}y_i' = \frac{p_i'q_i - p_iq_i'+ p_iq_i}{q_i^2}y_i.  \end{align*} So by induction, we deduce that for $0
      \le j < n$, \begin{align*} y_i^{(j)} = \frac{P_{i,j}}{q_i^j}y_i = \frac{q_i^{n-1-j}P_{i,j}}{q_i^{n-1}}y_i , \end{align*} where $q_i^{n-1 -j}P_{i,j}$ is a
      polynomial of degree at most $(n-1)d$.  Then, by the Leibniz's formula, $h_i^{(j)} = \sum_{k=0}^j \binom{j}{k}g_i^{(j-k)}y_i^{(j)}$, and the Wronskian has
      the form \begin{align*} W(\lst h n)    = \prod_{i=1}^n \frac{y_i}{q_i^{n-1}}\det M, \end{align*} with $M$ being a matrix whose entries are polynomials of
      degree at most $nd$. In particular, $\ord(\det M) \le \deg(\det M) \le n^2d$.      As $q_i(0) \neq 0$ for all $i$, we have: \begin{equation*} \ord (W(\mathbf{h})) = \sum_{i=1}^n \ord(y_i) + \ord(\det M) \le  \sum_{i=1}^n \ord(y_i) +
      n^2d.  \end{equation*} Finally, \cref{thm:upbd_wronS} implies the claimed bound.  \end{proof}

      As a direct corollary of \cref{thm:diffeq1}, we obtain upper bounds for the order of a sum of power series in $\mathbb{C}[[x]]$ in several different contexts. 

\begin{corollary} \label{coro:exp_proot_etc} Let $S(x) = \sum_{i=1}^n c_i g_i(x) y_i(x)$ be a non-zero sum, where $c_i \in \mathbb{C}$, and $g_i \in \mathbb{C}[x]$ are of degree at most $d$. Let $f_i \in \mathbb{C}[x]$ of degree at most $d$. 
\begin{enumerate}[(i)]
  \item \label{it:exp}If $y_i = \exp(f_i)$, then $\ord S \le n^2(d-1) + n-1$. 
        \item \label{it:sincos} If $y_i = \varphi(f_i)$ with $\varphi \in \{\cosh,\sinh, \cos , \sin\}$, then $\ord S \le 4n^2(d-1) + 2n-1$.  
        \item \label{it:rationalrealpower} If $y_i = \parc{\frac{p_i}{q_i}}^{\alpha_i}$, with $p_i, q_i \in \mathbb{C}[x]$ of degree at most $d$, $\alpha_i \in \mathbb{R}$ and $p_i(0), q_i(0) \neq 0$, then $\ord S
  \le 2 n^2 d + n -1$.  
\end{enumerate}
\end{corollary}

\begin{proof} For \cref{it:exp}, $y_i = \exp(f_i)$, $y_i$ admits a power series expansion and $\ord y_i = 0$ because $y_i(0) = \exp f_i(0) \neq 0$. Moreover, $y_i' - f_i'y_i = 0$,
    with $\deg f_i' \le d-1$. By \cref{thm:diffeq1}, $\ord S \le n^2(d-1) + n-1$.

    For \cref{it:sincos}, one can write $\cosh f= \frac{e^f + e^{-f}}{2}, \sinh f= \frac{e^f - e^{-f}}{2}$, $\cos f  = \frac{e^{\iu f} + e^{-\iu f}}{2},\sin f = \frac{e^{\iu f} - e^{-\iu f}}{2\iu}$. We can then apply~\cref{it:exp} with $2n$ terms in the sum and obtain the claimed bound.

    For \cref{it:rationalrealpower}, because both $p_i(0)$ and $q_i(0)$ are non-zero for all $i$, $y_i$ admits a power series expansion in $0$ and $\ord y_i = 0$. Moreover, it
satisfies $y_i' - \alpha_i \frac{p_i'q_i - p_i q_i'}{p_i q_i} y_i = 0$. Again, the bound is obtained by using ~\cref{thm:diffeq1}.  \end{proof}

A similar bound for $y_i = \exp f_i$ was already established in~\cite{voorhoeve1975wronskian} (see Example~1). The proof also relies on Wronskians. 

\section{ Sum of square roots of integers given by straight-line programs }\label{sec:sum_sqroot_slp}

As we mentioned in introduction, testing if an expression involving square roots is zero is an interesting problem and Bl{\"o}mer~\cite{blomer1993computing} developed
a polynomial time algorithm solving $\SSReq$. In $\SSReq$, the integers in input are given in binary and in this case, the problem is \emph{easy}. It is a natural extension of $\SSReq$ to ask what happens for the complexity of this problem in an algebraic model of computation, that is if the integers in input are given by straight line programs. More precisely, we want to study \cref{pb:SSRSLP}.

We would like to know if zero testing for this class of expressions is as easy as testing zero for straight line programs computing integers, namely $\equslp$
defined in~\cite{comp_num_analysis}. In the problem $\equslp$, given an input SLP $P$, we want to test if the integer $n_P$ computed by $P$ is zero. In~\cite{comp_num_analysis}, it is proven that $\equslp$ reduces to Polynomial Identity Testing ($\pit$), so it admits a
randomized polynomial time algorithm. In the hope of designing a randomized polynomial time algorithm for $\SSRslp$, we present in this section how one can
reuse some of the ideas of Bl{\"o}mer for $\SSReq$ in the context of \cref{pb:SSRSLP}. We show that the problem $\SSRslp$ can be \emph{reduced} under the
General Riemann Hypothesis (GRH) to a one dimensional case where all the square roots involved are on the same line over $\mathbb{Q}$. Actually, this special
one dimensional case captures all the hardness of the whole problem.  One can also find related results in~\cite{blomer1998probabilistic,balaji2022identity}.
They give randomized algorithms to decide if expressions involving radicals of depth 1 is equal to 0. In~\cite{balaji2022identity}, some algorithms are also
valid under GRH and use comparable arguments about density of certain prime numbers that we also develop in this section. However, the expressions involved in these
results are different; the expressions are not only sums of square roots but general arithmetic circuits involving square-roots of integers given in binary.

	\subsection{Reduction to the one-dimensional case}

			We now explain formally why it is enough to focus on the \emph{one dimensional} (\cref{pb:1dimSSR}) version of~\cref{pb:SSRSLP} in order to design an efficient
			algorithm for $\SSRslp$. We actually prove~\cref{thm:Turing_reduction}.

			\begin{theorem}
				\label{thm:Turing_reduction}
				Under GRH, there exists a randomized polynomial time Turing reduction from~\cref{pb:SSRSLP} to~\cref{pb:1dimSSR}.
			\end{theorem}

			The starting point is the following result due to Kneser \cite{Kneser1975} recalled by Bl{\"o}mer (\cite{blomer1993computing}, corollary~2.6).   

			\begin{lemma} Let $\lst a n$ be $n$ positive integers. Reals $\lsst a n$ are linearly dependent over $\mathbb{Q}$ if and only if there exist $ 1 \le i < j \le n
				$ such that $(\sqrt{a_i}, \sqrt{a_j})$ are linearly dependent over $\mathbb{Q}$ or equivalently \begin{align} \frac{\sqrt{a_i}}{\sqrt{a_j}} \in \mathbb{Q}.
					\label{eq:lindep} \end{align} \label{lemma:pairwise_dependence} \end{lemma}

			\begin{corollary}
				\label{cor:subsum}
				Let $a_1, \dots, a_n$ be positive integers and $(\delta_1, \dots, \delta_n) \in \{-1,1\}^n$. Let $(\sqrt{a_{i_1}}, \dots, \sqrt{a_{i_\ell}})$ be a basis of $\spn_{\mathbb{Q}}(\sqrt{a_1},\dots, \sqrt{a_i})$. Then for all $1 \le i \le n$, there exists a unique $1 \le j \le \ell$ such that $\sqrt{a_i} \in \mathbb{Q} \cdot \sqrt{a_{i_j}}$. And \begin{align} \sum_{i=1}^n \delta_i \sqrt{a_i} = 0 \iff \forall 1 \le j \le \ell, \sum_{i:~\sqrt{a_i} \in \mathbb{Q} \sqrt{a_{i_j}}} \delta_i \sqrt{a_i} = 0. \end{align}
			\end{corollary}

			\begin{proof}
				Without loss of generality, we can assume $i_j = j$.
				Assuming existence, uniqueness is clear because if $\sqrt{a_i} \in \mathbb{Q} \cdot \sqrt{a_{j}}$ and $\sqrt{a_i} \in  \mathbb{Q} \cdot \sqrt{a_{j'}}$, then $\sqrt{a_{j}}$ and $\sqrt{a_{j'}}$ are not linearly independent. 
				Now, for an
				element of the basis existence is clear and for $\sqrt{a_i}$ not in the basis, $(\sqrt{a_{1}},\dots, \sqrt{a_{\ell}}, \sqrt{a_i} )$ is linearly dependent. By~\cref{lemma:pairwise_dependence}, two elements are linearly dependent, and it can not be between to elements of the basis. So $\sqrt{a_i}$ is involved in this pair.

				 We partition $\{1,\dots,n\} = I_1 \cup \dots \cup I_\ell$, 
				 with $I_j = \{i |~ \sqrt{a_i} = q_{i,j} \sqrt{a_{j}}, q_{i,j} \in \mathbb{Q}  \}$.
				Now, $\sum_{i=1}^n \delta_i \sqrt{a_i} = \sum_{j=1}^\ell  \left(\sum_{i \in I_j} \delta_i q_{i,j} \right) \sqrt{a_{j}}$. As $(\sqrt{a_{j}})_{1 \le j \le \ell}$ form a basis, the previous sum is zero if and only if for all $1 \le j \le \ell, \sum_{i \in I_j} \delta_i q_{i,j} = 0$. By multiplying by $\sqrt{a_{j}}$, we obtain
				\begin{align*} \sum_{i=1}^n \delta_i \sqrt{a_i} = 0 \iff \forall 1 \le j \le \ell, \sum_{i \in I_j} \delta_i \sqrt{a_i} = 0. \end{align*}
			\end{proof}
			The reduction for~\cref{thm:Turing_reduction} then works in two steps and can be found in~\cref{alg:general_scheme}. Given an instance of~\cref{pb:SSRSLP}, 
  we first build the partition of the set of square roots in $\ell$ one-dimensional subsets, and then use the oracle for~\cref{pb:1dimSSR} to test if each associated subsum is zero. To build the partition, \cref{alg:general_scheme} works as follows: for each integer $a_i$, either we have already seen an integer $a_j$ before such that $\sqrt{a_i}/\sqrt{a_j} \in \mathbb{Q}$ and we add $a_i$ to the same one dimensional sum as $a_j$ or we construct a new one for $a_i$.  
			To complete the proof of~\cref{thm:Turing_reduction}, it only remains to show that one can perform the test of~\cref{eq:lindep} efficiently under GRH.

			\begin{algorithm} \caption{Algorithm for \cref{pb:SSRSLP}}	\label{alg:general_scheme} \Input{$\lst a n$ integers given by SLPs, signs $\lst
					\delta n \in \{-1,1\}$.} \Output{Decide if $S = \sum_{i=1}^{n}\delta_i \sqrt{a_i} = 0$.}
				 SubSums $\leftarrow [~]$ \; \For{$1 \le i \le n$}{ $k \leftarrow
					0~; ~\mathrm{FoundSubSum} \leftarrow \mathrm{false}$ \; \While{$ k < \abs{\mathrm{SubSums}}$ and not $\mathrm{FoundSubSum}$}{ Pick an element $b$
						in $\mathrm{SubSums}[k]$ \; \eIf{$\sqrt{a_i} / \sqrt{b} \in \mathbb{Q}$}{ Add $a_i$ to $\mathrm{SubSums}[k]$ ; ~ $\mathrm{FoundSubSum} \leftarrow \mathrm{true}$ \; }{
							$k$++\; } } \textbf{if} not $\mathrm{FoundSubSum}$ \textbf{then} Add $\{a_i\}$ to SubSums \; } For each element of SubSums, test if the corresponding one
				dimensional sum is zero using oracle for~\cref{pb:1dimSSR}\; \end{algorithm}

				\begin{lemma} Let $a,b$ be two positive integers.
				$ \frac{\sqrt{a}}{\sqrt{b}} \in \mathbb{Q}$ iff $\sqrt{ab} \in \mathbb{N}$.
				\label{lemma:sqrt_prod} \end{lemma}
			\begin{proof}
				$\sqrt{a}/\sqrt{b} \in \mathbb{Q}$ is equivalent to the statement that for any prime $p$, the $p$-adic valuation $v_p(a/b)$ of $a/b$ is even, $v_p(a/b) = v_p(a) - v_p(b)$, or equivalently $v_p(a) + v_p(b)  = v_p(ab) $ is even, \ie $\sqrt{ab} \in \mathbb{N}$ . 
			\end{proof}

			\cref{lemma:sqrt_prod} suggests that for testing~\cref{eq:lindep}, we just have to check if $a_ia_j$ is a perfect square.  Note that if both $a_i$ and $a_j$ can be computed by SLPs of size $s$, then $a_ia_j$ admits an SLP of size $2s+1$. In the next section we design a randomized polynomial time algorithm to perform this task under GRH.

			\subsection{Testing if an SLP computes a perfect square}

			We now demonstrate an algorithm for the following~\cref{pb:perfect_square}.

			\begin{problem}
				\label{pb:perfect_square}
				Given an SLP of size $t$ computing a positive integer $a$, decide if $a$ is a perfect square.
			\end{problem}

			Let $a$ be a positive integer computed by an SLP of size $t$. By induction on $t \ge 0$, we have the bound $a \le 2^{2 ^t}$.

			Our algorithm is the following: sample at random a prime $ p \le  2^{q(t)}$, with  $q$ a polynomial to be determined later, compute $a \bmod p$ and test if $a \bmod p$  is a square in $\mathbb{F}_p$. All these computations can be done in polynomial time in $t$~\cite{antbachshallit1996}. 
			If $a$ is actually a perfect square, then for all such primes $p$, $a \bmod p$ is a square in $\mathbb{F}_p$ and the answer of the algorithm is correct.
			 Whereas if $a$ is not a perfect square, Chebotarev's density theorem guarantees that the set of primes $p$ for which the polynomial $ X^2 -a$ splits in $\mathbb{F}_p$, \ie, the set of primes $p$
			where $a$ is a square in $\mathbb{F}_p$ has density $\frac{1}{2}$~\cite{stevenhagen1996chebotarev}. To ensure that we can use small primes and keep a non negligible probability to have a correct answer, we need an
			effective version. This effective version of Chebotarev's density theorem requires GRH and can be found in~\cite{serre1981quelques} (Theorem~4).    

			Our application of the effective Chebotarev's density theorem leads us to the following key lemma. The statement directly follows from~\cite{serre1981quelques} (Theorem~4) for the splitting field of $X^2 - a$. This only makes sense when $\mathrm{disc}(X^2 -a) = 4a$ does not vanish in $\mathbb{F}_p$. This is the reason why we only consider primes $p$ that do not divide $4a$.
			\begin{lemma} Let $a$ be a positive integer that is not a square. For $x \ge 2$, we define:  \begin{align*} d_a(x) \eqdef \frac{ \left| \{ p \text{ prime } \le x , a
						\text{ is not a square mod $p$ and } p \nmid 4a \} \right| }{\left| \{ p \text{ prime } \le x \} \right|} .  \end{align*} There exists a constant $C$ such
				that, under GRH, for all $x \ge 2$, \begin{align*} \left| d_a(x) - \frac{1}{2} \right| \le C \frac{\log_2 x}{\sqrt{x}}\left(  \log_2 (4a) + 2 \log_2 x \right).
				\end{align*} \label{lemma:density_ineq} \end{lemma}

			Now we can state the main lemma. 
			\begin{lemma} Let $a \le 2^{2^t}$ be a positive integer that is not a perfect square. 
				Then there exists an integer polynomial $q$ such that for $x = 2^{q(t)}$, we have $d_a(x) \ge \frac{1}{4}$.  \label{lemma:density_prime} \end{lemma}

			\begin{proof} From \cref{lemma:density_ineq}, the following inequality holds: \begin{align*} d_a(x) \ge \frac{1}{2} - C \frac{\log_2 x}{\sqrt{x}} (\log_2(4a)+
					2\log_2 x).  \end{align*} To conclude, we need to ensure that \begin{align*} &C \frac{\log_2 x}{\sqrt{x}} (\log_2(4a)+ 2\log_2 x) \le \frac{1}{4},
					\text{\ie,} \\ &\log_2(4a) \le \frac{1}{4}\frac{\sqrt{x}}{C \log_2 x} - 2\log_2 x.  \end{align*} With $a \le 2^{2^t}$, and $x = 2^{q(t)}$, it is sufficient to have
				\begin{align}
					2^t +2 \le \frac{1}{4}\frac{2^{\frac{q(t)}{2}}}{C q(t)} - 2 q(t).
					\label{eq:condition_q}
				\end{align}
				and one can see that $q(t) = O(t^2)$ satisfies~\cref{eq:condition_q}.
\end{proof}

			If one chooses $q(t)$ as in~\cref{lemma:density_prime}, our randomized algorithm runs in polynomial time and has a one sided error for solving~\cref{pb:perfect_square}. This completes the proof of~\cref{thm:Turing_reduction}.

\subsection{About the One Dimensional Variant}

  In the previous section, we showed that it is enough to focus on~\cref{pb:1dimSSR}. Suppose $\spn_{\mathbb{Q}}(\lsst a n) =  \mathbb{Q}  \cdot \sqrt{a_1}$.
  Then, for all $i$, $\sqrt{a_ia_1} \in \mathbb{N}$. And $S =\sum_{i=1}^{n} \delta_i \sqrt{a_i} = 0 \iff S \sqrt{a_1} =\sum_{i=1}^{n} \delta_i \sqrt{a_i a_1} =
  0$, with $ S \sqrt{a_1} \in \mathbb{N}$. One approach  to design an algorithm to test $S \sqrt{a_1} = 0$, is reducing it modulo a randomly selected prime number $p$. 
  The problem with this approach, is that once we have reduced $a_ia_1$ modulo $p$, one can compute its two square roots in $\mathbb{F}_p$
  but there is no way to decide which one of the two is the correct representation of $\sqrt{a_1a_i}$ in $\mathbb{F}_p$. Determining the correct reduction of
  $S\sqrt{a_1}$ modulo $p$ seems like  a non trivial task. And we most likely need a different approach in order to tackle~\cref{pb:1dimSSR}.

  \section{An application to sums of logarithms}\label{sec:linearcomblogs}

  In order to motivate the interest of proving bounds for the order of sums of certain power series, we show that the approach of Kayal and
  Saha~\cite{kayal2012sum} that led to a non trivial statement for the Sum of Square Roots problem can also be used to establish non-trivial statements for
  other fundamental number theoretic problems. In particular in this section, we focus on the sums of logarithms that is a well-studied
  problem~\cite{baker1998logarithmic,etessami2014note} and is related to deep number theory conjectures. 

  \begin{problem} Given two lists $(\lst a n)$, $(\lst c n)$  of integers with $a_i > 0$, decide if \begin{align*} \sum_{i=1}^n c_i \log{a_i} > 0.
  \end{align*} \label{pb:log_form} \end{problem}

  Similarly as for the Sum of Square roots problem (\cref{pb:SSR}), the complexity in the bit-model of~\cref{pb:log_form} is an open question. A refinement of
  the $abc$-conjecture formulated by Baker~\cite{BakerWustholz1993} would imply~\cref{pb:log_form} to be in $\P$. Interesting details about the link between the
  complexity of~\cref{pb:log_form} and open questions in number-theory can be found in~\cite{etessami2014note}. All these conjectures essentially state the
  existence of a \emph{gap}, namely that a non-zero sum of logarithms can not be \emph{too close to zero}. These \emph{gaps} or lower bounds are very similar
  to~\cref{pb:SSR_lower_bd} but in the case of logarithms. Our goal is to reuse the analogy between polynomials and integers in order to deduce a lower bound
  for a non trivial class of instances of~\cref{pb:log_form} via a lower bound on the order of a non-zero sum of logarithms of real polynomials.

\sumoflogsorderboundprop*
\begin{proof}
    First, if $f_i(0) > 0$, the power series $\log f_i(x)$ is well defined in a neighborhood of 0. So, $S$ admits a valid power series expansion.
    Now, we have \begin{align*} S'(x) = \sum_{i=1}^n c_i \frac{f_i'(x)}{f_i(x)} = \frac{\sum_{i=1}^n c_i F_i(x)}{\prod_{i=1}^n f_i(x)}, \end{align*} with
    $F_i(x) = f_i'(x) \prod_{j \neq i} f_j(x)$, is a polynomial of degree $\le (n-1)d + d-1 = nd -1 $. As $f_i(0) \neq 0$, we have:
    \begin{align*} \ord(S') = \ord\parc{\sum_{i=1}^n c_i F_i(x)} \le \deg \parc{\sum_{i=1}^n c_i F_i(x)} \le nd - 1.  
    \end{align*} Finally,
    \begin{align*} 
      \ord(S) \le \ord(S') +1 \le nd.  
    \end{align*} 
  \end{proof}

     And now, we can instantiate~\cref{prop:log} over integers and obtain the following gap theorem for a restricted but non trivial class of the instances
    of~\cref{pb:log_form}.

\sumoflogsofpolyintsthm*
\begin{proof}
          We choose any integer polynomials $p_1, p_2$ satisfying the following conditions and show that they satisfy the claimed bounds:
          \begin{align}
              p_1(n,d) &\ge 20 dn \log(dn), \label{eq:bd_p1}\\
              p_2(n,d) &\ge 1 + dn(\log(dn) + 1). \label{eq:bd_p2}
          \end{align}
          One can write $E$ as follows: \begin{align*} E &=  \underbrace{ \sum_{i=1}^n c_i d_i \log X }_{= A} + \underbrace{ \sum_{i=1}^n c_i \log
          \left(  1  + \frac{b_{1,i}}{X} + \dots + \frac{b_{d_i,i}}{X^{d_i}} \right)  }_{= S}.  \end{align*}
          \textbf{If $\mathbf{A \neq 0}$}, as $\sum_{i=1}^n c_i d_i$ is a non-zero integer, we have \begin{align*} \abs{A} \ge \log X.
    \end{align*} So an upper bound for $S$ of the form $\abs{S} \le \frac{1}{2} \log X$ is enough to prove that $\abs E \ge \frac{1}{2}\log X$. And we have the bound \begin{align*} \abs{S} &\le n C \max_{1 \le i \le n } \log \left(  1  + \frac{b_{1,i}}{X} + \dots +
              \frac{b_{d_i,i}}{X^{d_i}} \right).  \end{align*}
Now, as $\forall x\ge 0 ,\log (1+x) \le x$, and $\frac{b_{1,i}}{X} + \dots + \frac{b_{d_i,i}}{X^{d_i}} \le \frac{B}{X -1}$, we have
                  \begin{align*} \abs{S} &\le n C \frac{B}{X - 1 }.  \end{align*}
But, as $C \le X^{\frac{1}{2}}$ 
                  \begin{align*} \left| S \right|  \le nC \frac{B}{X - 1 } &\le 
                   \frac{nB X^{\frac{1}{2}}}{X - 1} \le 
                  \frac{2nB}{X^{\frac{1}{2}}}\le
                  \frac{2nB}{(B+1)^{\frac{1}{2}p_1(n,d) }},  \end{align*}
                  because $X -1 \ge \frac{X}{2}$ since $X \ge 2$.
Then, to have $\left| S \right| \le \frac{1}{2}\log X$, it is sufficient that: \begin{align} p_1(n,d)\log(B+1)
                  (B+1)^{\frac{1}{2}p_1(n,d) } - 4nB \ge 0. 
                  \label{eq:caseAneq0} \end{align}
One can easily check that~\cref{eq:caseAneq0} holds using~\cref{eq:bd_p1} and thus, $\left|E\right| \ge \frac{1}{2}\log X \ge \frac{1}{X^{p_2(n,d)}} $.

\textbf{If $\mathbf{A = 0}$}, we show that $\left| E \right| = \left| S \right|   \ge  \frac{1}{X^{p_2(n,d)}} $. Let
                  $y = \frac{1}{X}$, we have \begin{equation*} S = \sum_{i=1}^n c_i \underbrace{\log \left( 1+ b_{1,i} y + \dots + b_{d_i,i} y^{d_i} \right)
                    }_{= h_i(y)} = \sum_{j \ge 1 } \underbrace{y^j \sum_{i=1}^n c_i h_{i,j}}_{S_j} \end{equation*} \begin{equation*} \text{with }h_i(y) =
                    \sum_{j=1}^\infty h_{i,j}y^j = \sum_{k=1}^\infty  \frac{(-1)^{k+1}}{k} \left( b_{1,i} y + \dots + b_{d_i,i} y^{d_i} \right)^k.
                  \end{equation*}
By applying \cref{prop:log}, the minimum exponent $\ell$ such that $S_\ell \neq 0$ is such that $\ell \le dn$. The idea of the proof is to show that
                  for every $t \ge 1$, \begin{align} \frac{\abs{S_{\ell+t}}}{\abs{S_\ell}} \le \frac{1}{2^{t+1}}, \label{eq:goal} \end{align}
because in this case, we would have \begin{align} \abs{S} \ge \abs{\abs{S_\ell}  - \abs{S_{\ell+1}} - \dots}
                  \ge \abs{\abs{S_\ell}- \frac{1}{2} \abs{S_\ell}} = \frac{1}{2} \abs{S_\ell}, \label{eq:lowerbound} \end{align} and
                  one can obtain a lower bound  for $\abs{S}$ via a lower bound for $\abs{S_\ell}$. But to satisfy \cref{eq:goal}, one also needs an upper bound on $S_{\ell+t}$, for every $t$.

                  These two lower and upper bounds are both obtained as in the proof of~\cref{thm:SSR_PolynomialInteger} in~\cite{kayal2012sum}. The technical
                  details can be found in~\cref{app:log_proof}. In the end, we show that conditions on $p_1$ and $p_2$ (\cref{eq:bd_p1,eq:bd_p2}) are sufficient in order to satisfy~\cref{eq:goal,eq:lowerbound} and then deduce that $\left|E \right| = \left| S\right| \ge \frac{1}{X^{p_2(n,d)}}$.

                  \end{proof}

                  \subsection{Bounds for the Proof of~\cref{thm:sumlog}} \label{app:log_proof} We recall that we have

                  \begin{equation*} S = \sum_{i=1}^n c_i \underbrace{\log \left( 1+ b_{1,i} y + \dots + b_{d_i,i} y^{d_i} \right) }_{= h_i(y)} = \sum_{j \ge 1 }
                    \underbrace{y^j \sum_{i=1}^n c_i h_{i,j}}_{S_j} \end{equation*} \begin{equation*} \text{with }h_i(y) = \sum_{j=1}^\infty h_{i,j}y^j =
                  \sum_{k=1}^\infty  \frac{(-1)^{k+1}}{k} \left( b_{1,i} y + \dots + b_{d_i,i} y^{d_i} \right)^k.  \end{equation*}

      The goal is to show that $p_1(n,d)$ and $p_2(n,d)$ satisfying~\cref{eq:bd_p1,eq:bd_p2} give both an upper bound and a lower bound on $S_j$ that allows us to satisfy~\cref{eq:goal}.
      \newline \textbf{An upper bound on $S_j$:} \newline For any $j$, $h_{i,j}$ is contributed by the term of order $j$ in the sum $\sum \frac{(-1)^{k+1}}{k} \left( b_{1,i} y +
    \dots + b_{d_i,i} y^{d_i} \right)^k $ for $k$ in range $\left[1, j\right]$. Then \begin{align*} h_{i,j} &= \sum_{k=1}^j \frac{(-1)^{k+1}}{k} v_{k,i,j},
      \text{ with } \\ v_{k,i,j} &= \sum_{\substack{k_1 + 2 k_2 + \dots + d_i k_{d_i} = j \\ \lstp{k}{d_i} = k }} \binom{k}{\lst{k}{d_i}}
    b_{1,i}^{k_1} \dots b_{d_i,i}^{k_{d_i}}.  \end{align*}
where $\lst k {d_i}$ are non-negative integers. Notice that for $k < \frac{j}{d_i}, v_{k,i,j}= 0$ as the sum is empty. Then \begin{align*} \abs{v_{k,i,j}} &\le \sum_{\lstp{k}{d_i} = k } \binom{k}{\lst{k}{d_i} } \left| b_{1,i} \right|^{k_1} \dots \left| b_{d_i,i} \right|^{k_{d_i}} \\ & \le (B d_i)^k \le (B d )^k \tag{Multinomial theorem}.  \end{align*}
Since $\left| \frac{(-1)^{k+1}}{k} \right| \le 1$, \begin{align*} \left| h_{i,j}\right| \le \sum_{k = 0 }^j (Bd)^k \le (Bd)^{j+1}.  \end{align*}
Hence, let $S = \sum_{j \ge 1} S_j$, with $S_j = y^j \sum_{i=1}^n c_i h_{i,j}$ satisfies for all $j \ge 1$ \begin{align*} \left| S_j \right| \le y^j n
            C (Bd)^{j+1}.  \end{align*} \newline \textbf{A lower bound on $S_\ell$:} \newline 
            Now we prove a lower bound for $\left| S_\ell \right| = y^\ell \left| \sum_{i=1}^n
          c_i h_{i,\ell}\right| \neq 0$. We have \begin{align*} \sum_{i=1}^n c_i h_{i,\ell} &= \sum_{i=1}^n c_i \sum_{k = 1}^{\ell} \frac{(-1)^{k+1}}{k} v_{k,i,\ell} = \sum_{k = 1}^{\ell}
           \frac{(-1)^{k+1}}{k}  \underbrace{\left( \sum_{i=1}^n c_i v_{k,i,\ell}  \right)}_{\in \mathbb{Z}}.  \end{align*}
So $\ell! \left( \sum_{i=1}^n c_i h_{i,\ell} \right)$ is an integer. Then, if $\sum_{i=1}^n c_i h_{i,\ell} \neq 0$, 
            \begin{align*} \ell! \left| \sum_{i=1}^n c_i h_{i,\ell} \right| \ge
            1. \end{align*} And finally, we obtain \begin{align*} \left| S_\ell \right| \ge \frac{y^\ell}{\ell!}. \end{align*}
		With the obtained upper and lower bound, we can now deduce a requirement that ensures \cref{eq:goal} to be satisfied. Actually, it is sufficient
            that for every $t \ge 1$, \begin{align} ny^t C (Bd)^{\ell+t+1} \ell! \le \frac{1}{2^{t+1}}, \ie \\ X^t \ge 2^{t+1}n C (Bd)^{\ell+t+1} \ell!.
            \label{eq:bound_log_X} \end{align} 
            Recall that $C \le X^{1/2} $ and $\ell \le dn$, so it suffices that 
            \begin{align*}
                X^{t - \frac{1}{2}} \ge 2^{t+1}n(Bd)^{nd + t +1}(nd)!.   
            \end{align*}
            By applying $\log$, it is sufficient that for all $t \ge 1$,
            \begin{align*}
                (t- \frac{1}{2})p_1(n,d)\log(B+1) \ge &(t+1)\log(2) + \log(n)   \\& +(nd + t+1)\log(Bd)  + nd \log(nd). 
            \end{align*}
            Therefore, one can show that $p_1(n,d) \ge 20 dn \log(dn)  $ is sufficient (this bound is
            not optimized at all). And in this case, by \cref{eq:lowerbound}, we have

            \begin{align*} \left| S \right| \ge \frac{1}{2}  \left| S_\ell \right| \ge \frac{1}{2 \ell! X^\ell} \ge \frac{1}{2^{\ell\log \ell +1} X^\ell} \ge
            \frac{1}{X^{p_2(n,d)}}, \end{align*} with $p_2(n,d) \ge dn (\log (dn)+ 1 ) +1 $.

            \section{Conclusion and Open Questions}
In ~\cref{thm:diffeq1}, we greatly generalized the ideas developed in \cite{kayal2012sum}. This leads to the fact that the results in~\cite{kayal2012sum} can be generalized to various other families of power series, as evident in~\cref{coro:exp_proot_etc}. The Wronskian approach of~\cite{kayal2012sum} was also proven to be tight, up to a polynomial blowup. Our work also made significant progress in the straight line variant of the classical SSR equality problem. We also extended the results of~\cite{kayal2012sum} for sign testing of the sum of square roots of polynomial integers to a linear combination of logarithms of polynomial integers, solving a special case of Problem 2 proposed in~\cite{etessami2014note}. There are several directions for future research:
            \begin{enumerate}
              \item \cref{it:sincos} in \cref{coro:exp_proot_etc} applies to some special cases of solutions of 2nd order differential equations, like $\cos f$ where $f$ is a polynomial. It would be intriguing to see if this applies to generic solutions of higher order equations. 
              \item Can we generalize \cref{thm:order_SSR} by bounding the order of sums of arbitrary algebraic power series? Square roots of polynomials are examples of algebraic functions. However, finding a satisfying way to bound the order of the Wronskian determinant of a family of algebraic power series is challenging. The best upper bound we could find is about $r^{\mathrm{poly}(n)}d$ where $(r,d)$ is a bound on the bi-degree of the polynomial equations satisfied by the algebraic power series. 
\item We saw that both the sign testing of the sum of square roots and logarithms are specific examples of $\posslp$. In order to advance in understanding the complexity of $\posslp$, it is important to consider other special cases. For instance, the following are additional special cases of $\posslp$ that are worth exploring. 
            \begin{enumerate}
              \item Given positive integers $a,b,c,n$ in binary, determine the sign of $a^n+b^n-c^n$. To our knowledge, even this very special case of $\posslp$ remains open.
              \item Given an integer trinomial $f(x)\eqdef c_1 + c_2x^a_2 + c_3x^{a_3}\in \Z[x]$ and a rational $\frac{p}{q}$, determine the sign of $f\parc{\frac{p}{q}}$. This problem was posed in~\cite{KOIRAN2019151}. Significant progress was made on it in~\cite{Boniface2022TrinomialsAD} but the general form of the problem remains open.
            \end{enumerate}
            \end{enumerate}
\section{Acknowledgements}
 Louis Gaillard expresses gratitude to Peter B\"urgisser for hosting him during an internship at TU Berlin in the summer of 2022, where this research work was conducted.
Gorav Jindal was a member of Graduiertenkolleg "Facets of Complexity/Facetten der Komplexit\"at" (GRK 2434) and Institut f\"ur Mathematik, Technische Universit\"at Berlin, where most of the work was completed.
            \bibliographystyle{alpha} \bibliography{bibliography}    \end{document}